\pdfoutput=1
\documentclass[conference]{IEEEtran}
\usepackage{mathtools} % \allowdisplaybreaks, \coloneqq, \binom, ...
\usepackage{amsthm, amssymb} % theorems and symbols
\usepackage{bm} % bold vectors
\usepackage{microtype} % micro-typographic adjustments
\usepackage[noadjust]{cite}
\usepackage[hidelinks]{hyperref} % cross-referencing and hypertext links, bookmarks, ...
\usepackage[nameinlink]{cleveref} % references
\usepackage{orcidlink} % ORCID symbol

\usetikzlibrary{calc}

\hypersetup{hidelinks, pdfinfo = {Title = {The Exact Load-Memory Tradeoff of Multi-Access Coded Caching With Combinatorial Topology}, Author = {Federico Brunero and Petros Elia}, Subject = {Information-Theoretic Caching}, Keywords = {Coded Caching, Combinatorial Topology, Index Coding, Information-Theoretic Converse, Multi-Access Coded Caching (MACC)}}}

\allowdisplaybreaks

\Crefformat{figure}{#2Fig.~#1#3}

\theoremstyle{plain}
\newtheorem{theorem}{Theorem}
\newtheorem{lemma}{Lemma}

\theoremstyle{definition}
\newtheorem{definition}{Definition}
\newtheorem{example}{Example}

\theoremstyle{remark}
\newtheorem{remark}{Remark}

\IEEEoverridecommandlockouts % to enable \thanks{} in conference mode

\title{The Exact Load-Memory Tradeoff of Multi-Access Coded Caching With Combinatorial Topology\thanks{This work was supported by the European Research Council (ERC) through the EU Horizon 2020 Research and Innovation Program under Grant 725929 (Project DUALITY). An extended version of this work~\cite{Brunero2021FundamentalLimitsCombinatorial} has been submitted to IEEE Transactions on Information Theory.}}

\author{\IEEEauthorblockN{Federico Brunero\textsuperscript{\orcidlink{0000-0002-6980-3827}} and Petros Elia\textsuperscript{\orcidlink{0000-0002-3531-120X}}}
\IEEEauthorblockA{Communication Systems Department, EURECOM, Sophia Antipolis, France\\
Email: \{brunero, elia\}@eurecom.fr}}

\begin{document}

\bstctlcite{IEEEexample:BSTcontrol}

\nocite{MaddahAli2014FundamentalLimitsCaching, Parrinello2020FundamentalLimitsCoded, Lampiris2021ResolvingFeedbackBottleneck, Toelli2020MultiAntennaInterference, Zhang2018CodedCachingArbitrary, Zhang2021DeepLearningWireless, Hachem2017CodedCachingMulti, Joudeh2021FundamentalLimitsWireless, Lampiris2020FullCodedCaching}

\maketitle

\begin{abstract}
    
    Recently, Muralidhar \emph{et al.} proposed a novel multi-access system model where each user is connected to multiple caches in a manner that follows the well-known combinatorial topology of combination networks. For such multi-access topology, the same authors proposed an achievable scheme, which stands out for the unprecedented coding gains even with very modest cache resources. In this paper, we identify the fundamental limits of such multi-access setting with exceptional potential, providing an information-theoretic converse which establishes, together with the inner bound by Muralidhar \emph{et al.}, the exact optimal performance under uncoded prefetching.
    
\end{abstract}


\section{Introduction}

Coded caching was introduced in~\cite{MaddahAli2014FundamentalLimitsCaching} as a coding-based communication technique with the purpose of reducing significantly the amount of data to be transferred from a centralized server to its cache-aided receiving users. The key idea behind coded caching involves an accurate joint design of the \emph{placement phase} and the \emph{delivery phase}. During the placement phase, which happens during off-peak hours, the caches of the users are preemptively filled without knowing the future requests. The caching is carefully performed so that, once the requests of the users are revealed, the amount of bits to be transferred during the delivery phase, which takes place when the network is saturated, is minimized. In the standard single-stream broadcast channel model with $K$ receiving users each of which is able to store in its cache a fraction $\gamma$ of the main library, coded caching is able to provide a sizeable coding gain equal to $K\gamma + 1$, where such coding gain represents simply the number of users to which a coded message is useful at the same time.

Current research on coded caching spans several topics such as the impact of multiple antennas on caching~\cite{Parrinello2020FundamentalLimitsCoded, Lampiris2021ResolvingFeedbackBottleneck, Toelli2020MultiAntennaInterference}, the interplay between caching and file popularity~\cite{Zhang2018CodedCachingArbitrary, Zhang2021DeepLearningWireless, Hachem2017CodedCachingMulti} and a variety of other scenarios~\cite{Joudeh2021FundamentalLimitsWireless, Lampiris2020FullCodedCaching}. For a thorough review of the existing coded caching works, we strongly encourage the reader to refer to the longer version of this work~\cite{Brunero2021FundamentalLimitsCombinatorial}.

\subsection{Multi-Access Coded Caching}

Differently from the model in~\cite{MaddahAli2014FundamentalLimitsCaching} where each user has access to its own single dedicated cache, it is conceivable that in several scenarios each cache serves more than one user, and that each user can connect to more than one cache. For instance, in dense cellular networks, the cache-aided access points (APs) could have overlapping coverage areas, hence allowing each user to connect to more than one AP. Such scenario motivated the work in~\cite{Hachem2017CodedCachingMulti}, where the authors introduced a new parameter $\lambda$ to keep into consideration the number of caches that each user can access. The corresponding model defined by $\lambda$ as well as by the number of users $K$, the number of library files $N$, and the cache size of $M$ files was referred to as the multi-access coded caching (MACC) model. Such model includes $\Lambda$ caches and $K = \Lambda$ users, where each of them is connected to $\lambda > 1$ consecutive caches in a cyclic wrap-around fashion.

Since the introduction of the aforementioned multi-access model with cyclic wrap-around topology, various works focused on the design of coding schemes that leverage the multi-access nature of the problem. For instance, a caching-and-delivery scheme was first proposed in the original work~\cite{Hachem2017CodedCachingMulti} with a decentralized (stochastic) cache placement, where this scheme provided improved local caching gains. Another achievable scheme preserving both the full local caching gain and the optimal coding gain of $\Lambda \lambda \gamma + 1$ was instead presented in~\cite{Serbetci2019MultiAccessCoded}, albeit for the rather unrealistically demanding scenario where $\lambda = (\Lambda - 1)/\Lambda\gamma$. Another notable work is then~\cite{Reddy2020RateMemoryTrade}, where the authors designed a novel scheme for any $\lambda \geq 1$, which was proved, for the similarly demanding regime of $\lambda \geq \Lambda/2$ and $\Lambda\gamma \leq 2$, to be at a factor of at most $2$ from the optimal under the assumption of uncoded placement. Other relevant works investigated the MACC problem and its connection to topics such as privacy and secrecy~\cite{Namboodiri2021MultiAccessCoded, Namboodiri2021MultiAccessCodeda}, structured index coding problems~\cite{Reddy2021StructuredIndexCoding} and PDA designs~\cite{Cheng2021NovelTransformationApproach}.

Recently, some interest arose towards new multi-access models which deviate from the cyclic wrap-around topology in~\cite{Hachem2017CodedCachingMulti}. For example, a new MACC paradigm was presented in~\cite{Katyal2021MultiAccessCoded}, which involved topologies that are inspired by cross resolvable designs (CRDs), a special class of designs in combinatorics. However, a substantial breakthrough came with the work in~\cite{Muralidhar2021MaddahAliNiesen}, where the authors proposed a MACC model enjoying the same amount of resources $\lambda$ and $\Lambda\gamma$, but where now the users and the caches are connected following the well-known combinatorial topology of combination networks~\cite{Ji2015FundamentalLimitsCaching}. This was a breakthrough because it allowed for the deployment of a subsequent scheme, presented in~\cite{Muralidhar2021MaddahAliNiesen} as a generalization of the original Maddah-Ali and Niesen (MAN) scheme in~\cite{MaddahAli2014FundamentalLimitsCaching}, that achieves the astounding coding gain $\binom{\Lambda\gamma + \lambda}{\lambda}$ far exceeding $\Lambda\gamma + 1$ even for small values of $\lambda$ and $\Lambda\gamma$, which is the regime that really matters.

\subsection{Main Contributions}

Our work explores the fundamental limits of this undoubtedly powerful MACC model by Muralidhar \emph{et al.}, which based on the findings in~\cite{Muralidhar2021MaddahAliNiesen} has astounding performance. Such fundamental limits had remained entirely unknown as no information-theoretic converse has ever been developed. We here establish the exact optimal performance with the introduction of a novel information-theoretic lower bound on the optimal worst-case communication load under the assumption of uncoded placement. Further results are presented in~\cite{Brunero2021FundamentalLimitsCombinatorial}, such as a generalization of the achievable scheme in~\cite{Muralidhar2021MaddahAliNiesen} as well as novel information-theoretic converses for the topology-agnostic multi-access setting, i.e., the setting where it is not known a priori how the $K$ users are connected to the $\Lambda$ caches in the system.

\subsection{Paper Outline}

The paper is organized as follows. The MACC model and some preliminary definitions are presented in \Cref{sec: System Model}. \Cref{sec: Main Result} presents the information-theoretic converse, whereas its general proof is described in~\Cref{sec: Proof of the Converse Bound}. \Cref{sec: Conclusions} concludes the paper.

\subsection{Notation}

We denote by $\mathbb{Z}^{+}$ the set of positive integers. % and by $\mathbb{Z}^{+}_{0}$ the set of non-negative integers.
For $n \in \mathbb{Z}^{+}$, we define $[n] \coloneqq \{1, \dots, n\}$. If $a, b \in \mathbb{Z}^{+}$ such that $a < b$, then $[a : b] \coloneqq \{a, a + 1, \dots, b - 1, b\}$. For sets we use calligraphic symbols, whereas for vectors we use bold symbols. Given a finite set $\mathcal{A}$, we denote by $|\mathcal{A}|$ its cardinality. We denote by $\binom{n}{k}$ the binomial coefficient and we let $\binom{n}{k} = 0$ whenever $n < 0$, $k < 0$ or $n < k$.

\section{System Model}\label{sec: System Model}

We consider the centralized coded caching scenario where one single server has access to a library $\mathcal{L} = \{W_{n} : n \in [N]\}$ containing $N$ files of $B$ bits each. The server is connected to $K$ users through an error-free broadcast link. In the system there are $\Lambda$ caches, each of size $MB$ bits. In agreement with the system model in~\cite{Muralidhar2021MaddahAliNiesen}, each user is connected \emph{exactly} and \emph{uniquely} to a subset of $\lambda$ caches for some fixed value of $\lambda \in [\Lambda]$, which consequently implies that there are $K = \binom{\Lambda}{\lambda}$ users for any given $\Lambda$ and $\lambda \in [\Lambda]$. We denote\footnote{For ease of notation, we will often omit braces and commas when indicating sets. For instance, user $\{1, 2\}$ in~\Cref{fig: MACC Model Example} is denoted as $12$.} by $\mathcal{U} \subseteq [\Lambda]$ the user connected to the $|\mathcal{U}| = \lambda$ caches in the set $\mathcal{U}$. We further assume that the link between the server and the users is the main bottleneck, whereas we assume that the channel between each user and its assigned caches has infinite capacity. As is common, we assume that $N \geq K$. Such setting is completely described by the tuple $(\Lambda, \lambda, N, M)$ and we refer to it as the MACC problem with combinatorial topology.

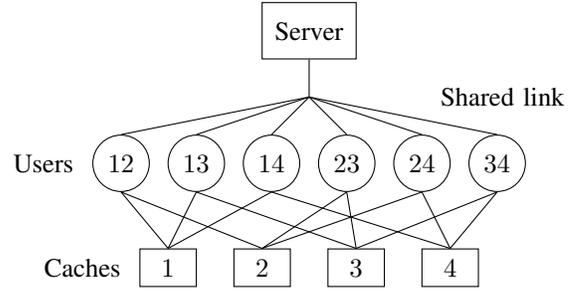
\begin{figure}
    \centering
    \tikzset{cnode/.style={draw, circle, inner sep = 0, minimum size = 0.75cm}}
    \tikzset{snode/.style={draw, rectangle, text centered}}
    \begin{tikzpicture}
        \node[snode, minimum height = 0.75cm, minimum width = 1.25cm](server){Server};
        \draw (server.south)--++(0, -0.5);
        \foreach \a/\b/\c/\n in {1/1/2/-2.5, 2/1/3/-1.5, 3/1/4/-0.5, 4/2/3/0.5, 5/2/4/1.5, 6/3/4/2.5}{
        \draw ($(server.south) + (0, -0.5)$)--++(\n, -0.5) node[cnode, below](user\a){$\b\c$};
        }
        \node at (user1.west)[left, xshift = -0.125cm]{Users};
        \node at ($(server.south) + (1.5, -0.5)$)[right, xshift = 0.125cm]{Shared link};
        \node at ($(user1.south) + (0.625, -0.75)$)[snode, below, minimum height = 0.5cm, minimum width = 0.75cm](cache1){$1$};
        \node at ($(user1.south) + (1.875, -0.75)$)[snode, below, minimum height = 0.5cm, minimum width = 0.75cm](cache2){$2$};
        \node at ($(user1.south) + (3.125, -0.75)$)[snode, below, minimum height = 0.5cm, minimum width = 0.75cm](cache3){$3$};
        \node at ($(user1.south) + (4.375, -0.75)$)[snode, below, minimum height = 0.5cm, minimum width = 0.75cm](cache4){$4$};
        \draw (user1.south)--(cache1.north); \draw (user1.south)--(cache2.north);
        \draw (user2.south)--(cache1.north); \draw (user2.south)--(cache3.north);
        \draw (user3.south)--(cache1.north); \draw (user3.south)--(cache4.north);
        \draw (user4.south)--(cache2.north); \draw (user4.south)--(cache3.north);
        \draw (user5.south)--(cache2.north); \draw (user5.south)--(cache4.north);
        \draw (user6.south)--(cache3.north); \draw (user6.south)--(cache4.north);
        \node at (cache1.west)[left, xshift = -0.125cm]{Caches};
    \end{tikzpicture}
    \caption{MACC problem with combinatorial topology and $\Lambda = 4$ caches, where each user is connected exactly and uniquely to a subset of $\lambda = 2$ caches.}
    \label{fig: MACC Model Example}
\end{figure}

The communication procedure proceeds as follows. During the placement phase, which typically occurs well before the delivery phase, the central server fills the caches without any knowledge of the future requests from the users. The delivery phase, which typically happens when the network is saturated and interference-limited, commences when the file-requests of all users are simultaneously revealed. During delivery, the server prepares some coded messages which are sent over the bottleneck shared link, so that each user can retrieve the missing information from the received transmission. The users cancel the interference terms that appear in the broadcast transmission, and do so by means of the cached contents they have access to, eventually decoding their own messages. As a consequence, the worst-case communication load $R$ is defined as the total number of transmitted bits, normalized by the file-size $B$, that can guarantee the correct delivery of any $K$-tuple of requested files in the worst-case scenario. The optimal communication load $R^\star$ is formally defined as
\begin{equation}
    R^\star(M) \coloneqq \inf\{R : \text{ $(M, R)$ is achievable}\}
\end{equation}
where the tuple $(M, R)$ is said to be \emph{achievable} if there exists a caching-and-delivery procedure for which, for any possible demand, a load $R$ can be guaranteed for a given memory value $M$.

We use the notation $W_{d_{\mathcal{U}}}$ to denote the file requested by the user identified by $\mathcal{U}$ for some $\mathcal{U} \subseteq [\Lambda]$ with $|\mathcal{U}| = \lambda$, which we remind the reader is simply the user connected exactly and uniquely to the $\lambda$ caches in the set $\mathcal{U}$. For the sake of simplicity, we denote by $\bm{d} = (d_{\mathcal{U}} : \mathcal{U} \subseteq [\Lambda], |\mathcal{U}| = \lambda)$ the demand vector containing the indices of the files requested by the users in the system, i.e., $d_{\mathcal{U}} \in [N]$ for each $\mathcal{U} \subseteq [\Lambda]$ with $|\mathcal{U}| = \lambda$. The following example can help familiarize the reader with the setting.

\begin{example}[$\Lambda = 4, \lambda = 2, N, M$]
    Consider the MACC problem with combinatorial topology and $\Lambda = 4$ caches in~\Cref{fig: MACC Model Example}. Each set of $\lambda = 2$ caches is uniquely assigned to a user, hence there are $K = \binom{\Lambda}{\lambda} = 6$ users in total. Recalling that each user is identified by the set of $2$ caches it is connected to and that for simplicity we omit braces and commas when indicating sets, we let $12$ represent the user connected to cache~$1$ and cache~$2$, we let $13$ represent the user connected to cache~$1$ and cache~$3$, and so on. The demand vector is given by $\bm{d} = (d_{12}, d_{13}, d_{14}, d_{23}, d_{24}, d_{34})$, where $d_{12} \in [N]$ is the index of the file requested by the user $12$, $d_{13} \in [N]$ is the index of the file requested by the user $13$, and so on.
\end{example}

Our goal is to provide a converse bound on the optimal worst-case load under the assumption of uncoded placement, whose definition is given in the following.
\begin{definition}[Uncoded Cache Placement]\label{def: Uncoded Cache Placement}
    A cache placement is \emph{uncoded} if the bits of the files are simply copied within the caches of the users.
\end{definition}

Denoting by $R^\star_{\text{u}}$ the optimal worst-case load under uncoded placement, it trivially holds $R^\star_{\text{u}} \geq R^\star$.

\section{Main Result}\label{sec: Main Result}

We present the main contribution of the paper. The derivation of the converse bound adopts the index coding technique first proposed in~\cite{Wan2020IndexCodingApproach}. Our main challenge will be to design the converse in such a way that it tightly captures the fact that each user, differently from the original setting in~\cite{MaddahAli2014FundamentalLimitsCaching}, is connected to $\lambda$ caches for some $\lambda \in [\Lambda]$. The result is stated in the following theorem.

\begin{theorem}\label{thm: Converse Bound Theorem}
    Consider the multi-access coded caching problem with combinatorial topology with parameters $(\Lambda, \lambda, N, M)$. Under the assumption of uncoded cache placement, the optimal worst-case communication load $R^\star_{\textnormal{u}}$ is a piecewise linear curve with corner points
    \begin{equation}
        (M, R^\star_{\textnormal{u}}) = \left(t \frac{N}{\Lambda}, \frac{\binom{\Lambda}{t + \lambda}}{\binom{\Lambda}{t}} \right), \quad \forall t \in [0 : \Lambda].
    \end{equation}
\end{theorem}

\begin{proof}
    The proof of achievability comes from the coding scheme in~\cite{Muralidhar2021MaddahAliNiesen}, whereas the proof of the converse bound is presented in \Cref{sec: Proof of the Converse Bound}.
\end{proof}

The longer version in~\cite{Brunero2021FundamentalLimitsCombinatorial} of this work provides also another interesting result that we will simply mention here. The system model in~\cite{Muralidhar2021MaddahAliNiesen} considers a single value of $\lambda \in [\Lambda]$ for a given $\Lambda$, although such setting can be extended into a \emph{generalized} combinatorial multi-access model which supports the coexistence of groups of users connected to different numbers of caches. This extension can be obtained if we consider a simultaneous coexistence of each model in~\cite{Muralidhar2021MaddahAliNiesen} for each $\lambda \in [\Lambda]$. The quite surprising outcome in~\cite{Brunero2021FundamentalLimitsCombinatorial} is summarized in the following remark.

\begin{remark}
    For the MACC problem with generalized combinatorial topology, there is no need to encode over users which are connected to different numbers of caches, even though there is abundance of coding opportunities. Instead, applying the coding scheme in~\cite{Muralidhar2021MaddahAliNiesen} in a TDMA-like manner is sufficient to achieve the optimal communication load under uncoded prefetching. An additional powerful insight that comes out of this is that the basic $\Lambda$-cache MAN placement proves to be extremely effective as it allows for the optimal performance for any instance of the generalized combinatorial topology.
\end{remark}

We now proceed by presenting the converse proof of \Cref{thm: Converse Bound Theorem}.

\section{Proof of the Converse Bound}\label{sec: Proof of the Converse Bound}

The converse relies on the well-known acyclic subgraph index coding bound, which has been extensively used in various other settings (see for example~\cite{Wan2020IndexCodingApproach, Parrinello2020FundamentalLimitsCoded, Brunero2021UnselfishCodedCaching} to name a few) in order to derive lower bounds on the optimal worst-case load in caching under the assumption of uncoded prefetching. To clarify the connection between this bound and our setting, we start by providing a brief presentation of the index coding problem and its connection to coded caching.

\subsection{Definition of the Index Coding Problem}

The index coding problem \cite{BarYossef2011IndexCodingSide} consists of a central server having access to $N'$ independent messages, and of $K'$ users that are connected to the server via a shared error-free broadcast channel. Each user $k \in [K']$ has a set of desired messages $\mathcal{M}_k \subseteq [N']$, which is called the \emph{desired message set}, while also having access to another subset of messages $\mathcal{A}_k \subseteq [N']$, which is called the \emph{side information set}. To avoid trivial scenarios, it is commonly assumed that $\mathcal{M}_k \neq \emptyset$, $\mathcal{A}_k \neq [N']$ and $\mathcal{M}_k \cap \mathcal{A}_k = \emptyset$ for each $k \in [K']$.

The index coding problem is usually described in terms of its \emph{side information graph}. Let $M_i$ be the $i$-th message for some $i \in [N']$. Then, such graph is a directed graph where each vertex is a desired message and where there exists an edge from a desired message $M_i$ to a desired message $M_j$ if and only if message $M_i$ is in the side information set of the user requesting message $M_j$. The derivation of our converse is based on the following bound from~\cite[Corollary 1]{Arbabjolfaei2013CapacityRegionIndex}.

\begin{lemma}[{\cite[Corollary 1]{Arbabjolfaei2013CapacityRegionIndex}}]\label{lem: Acyclic Subgraph Converse Bound}
    Consider an index coding problem with $N'$ messages $M_i$ for $i \in [N']$. The minimum number of transmitted bits $\rho$ is lower bounded as
    \begin{equation}
        \rho \geq \sum_{i \in \mathcal{J}}|M_i|
    \end{equation}
    for any acyclic subgraph $\mathcal{J}$ of the side information graph.
\end{lemma}

\subsection{Main Proof}

The first step in our converse proof consists of dividing, in the most generic manner, each file into a maximum of $2^\Lambda$ disjoint subfiles as
\begin{equation}\label{eqn: File Splitting}
    W_{n} = \left\{W_{n, \mathcal{T}} : \mathcal{T} \subseteq [\Lambda] \right\}, \quad \forall n \in [N]
\end{equation}
where we identify with $W_{n, \mathcal{T}}$ the subfile which is exclusively stored by the caches in $\mathcal{T}$. Noticing that the bits of the library files are simply copied within the caches, such placement is uncoded according to \Cref{def: Uncoded Cache Placement}.

\subsubsection{Constructing the Index Coding Bound}

Assuming that each user requests a distinct\footnote{The set of worst-case demands may not include the set of demand vectors $\bm{d}$ with all distinct entries. However, our goal is to derive a converse bound on the worst-case load, hence this is not a problem. Indeed, the choice of treating distinct demands yields a valid converse bound, since such bound does not need to be, a priori, the tightest bound. In our case, it proves to be tight.} file, we consider the index coding problem with $K' = K = \binom{\Lambda}{\lambda}$ users and $N' = \binom{\Lambda}{\lambda} 2^{\Lambda - \lambda}$ independent messages, where each such message represents a subfile requested by some user (who naturally does not have access to it via a cache). Recalling that $W_{d_{\mathcal{U}}}$ denotes the file requested by the user identified by $\mathcal{U}$, the desired message set and the side information set are respectively given, in their most generic form, by
\begin{align}
    \mathcal{M}_{\mathcal{U}} & = \{W_{d_{\mathcal{U}}, \mathcal{T}} : \mathcal{T} \subseteq [\Lambda] \setminus \mathcal{U}\} \\
    \mathcal{A}_{\mathcal{U}} & = \{W_{n, \mathcal{T}} : n \in [N], \mathcal{T} \subseteq [\Lambda], \mathcal{T} \cap \mathcal{U} \neq \emptyset\}
\end{align}
for each user $\mathcal{U}$ with $\mathcal{U} \subseteq [\Lambda]$ and $|\mathcal{U}| = \lambda$. Here, the side information graph consists of a directed graph where each vertex is a subfile, and where there is an edge from the subfile $W_{d_{\mathcal{U}_1}, \mathcal{T}_1}$ to the subfile $W_{d_{\mathcal{U}_2}, \mathcal{T}_2}$ if and only if $W_{d_{\mathcal{U}_1}, \mathcal{T}_1} \in \mathcal{A}_{\mathcal{U}_2}$ with $\mathcal{U}_j \subseteq [\Lambda]$, $|\mathcal{U}_j| = \lambda$, $\mathcal{T}_j \subseteq [\Lambda] \setminus \mathcal{U}_j$ for $j \in \{1, 2\}$ and $\mathcal{U}_1 \neq \mathcal{U}_2$. Since our aim is to apply \Cref{lem: Acyclic Subgraph Converse Bound}, we need to consider acyclic sets of vertices $\mathcal{J}$ in the side information graph. Toward this, we take advantage of the following lemma.

\begin{lemma}\label{lem: Acyclic Lemma}
    Let $\bm{d} = (d_{\mathcal{U}} : \mathcal{U} \subseteq [\Lambda], |\mathcal{U}| = \lambda)$ be a demand vector and let $\bm{c} = (c_1, \dots, c_\Lambda)$ be a permutation of the $\Lambda$ caches. The following set of vertices
    \begin{equation}
        \bigcup_{i \in [\lambda : \Lambda]}  \bigcup_{\substack{\mathcal{U}^i \subseteq \{c_1, \dots, c_i\} : |\mathcal{U}^i| = \lambda,\\ c_i \in \mathcal{U}^i}} \bigcup_{\mathcal{T}_i \subseteq [\Lambda] \setminus \{c_1, \dots, c_i\}} \left\{ W_{d_{\mathcal{U}^i}, \mathcal{T}_i} \right\}
    \end{equation}
    is acyclic.
\end{lemma}

\begin{proof}
    Due to lack of space, the proof is relegated to the longer version of this work~\cite{Brunero2021FundamentalLimitsCombinatorial}, where the proof is provided for a generalization of the lemma here presented.
\end{proof}

Consider a demand vector $\bm{d}$ and a permutation $\bm{c}$ of the set $[\Lambda]$. Applying \Cref{lem: Acyclic Lemma} yields the following lower bound
\begin{equation}\label{eqn: Index Coding Lower Bound 1}
    BR^\star_{\textnormal{u}} \geq R(\bm{d}, \bm{c})
\end{equation}
where $R(\bm{d}, \bm{c})$ is defined as
\begin{equation}
    R(\bm{d}, \bm{c}) \coloneqq \sum_{i = \lambda}^{\Lambda} \sum_{\substack{\mathcal{U}^i \subseteq \{c_1, \dots, c_i\} : |\mathcal{U}^i| = \lambda,\\ c_i \in \mathcal{U}^i}} \sum_{\mathcal{T}_i \subseteq [\Lambda] \setminus \{c_1, \dots, c_i\}} \left| W_{d_{\mathcal{U}^i}, \mathcal{T}_i} \right|.
\end{equation}

\subsubsection{Constructing the Optimization Problem}

Now our goal is to create several bounds as the one in~\eqref{eqn: Index Coding Lower Bound 1} considering any vector $\bm{d} \in \mathcal{D}$ and any vector $\bm{c} \in \mathcal{C}$, where we denote by $\mathcal{D}$ and $\mathcal{C}$ the set of possible demand vectors with distinct entries and the set of possible permutation vectors of the set $[\Lambda]$, respectively. Our aim is then to average all these bounds to obtain in the end a useful lower bound on the optimal worst-case load. Considering that $|\mathcal{D}| = \binom{N}{K}K!$ and $|\mathcal{C}| = \Lambda!$, we aim to simplify the expression given by
\begin{equation}\label{eqn: Complete Lower Bound 1}
    \binom{N}{K}K!\Lambda! BR^\star_{\textnormal{u}} \geq \sum_{\bm{d} \in \mathcal{D}} \sum_{\bm{c} \in \mathcal{C}} R(\bm{d}, \bm{c}).
\end{equation}
Toward simplifying~\eqref{eqn: Complete Lower Bound 1}, we proceed by counting how many times each subfile $W_{n, \mathcal{T}}$ --- for any given $n \in [N]$, $\mathcal{T} \subseteq [\Lambda]$ and $|\mathcal{T}| = t'$ for some $t' \in [0 : \Lambda]$ --- appears in~\eqref{eqn: Complete Lower Bound 1}.

Let us focus on the subfile $W_{n, \mathcal{T}}$ for some $n \in [N]$, $\mathcal{T} \subseteq [\Lambda]$ and $|\mathcal{T}| = t'$ with $t' \in [0 : \Lambda]$. Assume that the file $W_n$ is demanded by user $\mathcal{U}$ for some $\mathcal{U} \subseteq [\Lambda] \setminus \mathcal{T}$ with $|\mathcal{U}| = \lambda$ and denote by $\mathcal{D}_{n, \mathcal{U}}$ the set of demands such that $d_{\mathcal{U}} = n$. Out of the entire set $\mathcal{D}$ of all possible distinct demands, we find a total of $\binom{N}{K}K!/N$ distinct demands for which a file is requested by the same user, hence it holds $\left| \mathcal{D}_{n, \mathcal{U}} \right| = \binom{N}{K}K!/N$. For each $\bm{d} \in \mathcal{D}_{n, \mathcal{U}}$ and for each $\bm{c} \in \mathcal{C}$ there is a corresponding bound $R(\bm{d}, \bm{c})$. The subfile $W_{n, \mathcal{T}}$ appears only in the bounds induced by permutation vectors $\bm{c} \in \mathcal{C}$ such that the elements in the set $\mathcal{U}$ appear in the vector $\bm{c}$ before\footnote{Indeed, the subfile $W_{n, \mathcal{T}}$ appears in the acyclic graph chosen as in~\Cref{lem: Acyclic Lemma} for all those permutations $\bm{c} = (c_1, \dots, c_\Lambda)$ for which $\mathcal{U} = \mathcal{U}^i$ and $\mathcal{T} = \mathcal{T}_i$ for some $i \in [\lambda : \Lambda]$ such that $\mathcal{U}^i \subseteq \{c_1, \dots, c_i\}$ with $|\mathcal{U}^i| = \lambda$ and $c_i \in \mathcal{U}^i$, and such that $\mathcal{T}_i \subseteq \{c_{i + 1}, \dots, c_{\Lambda}\}$, i.e., this happens whenever the elements in $\mathcal{T}$ are after the elements in $\mathcal{U}$ in the permutation vector $\bm{c}$.} the elements in the set $\mathcal{T}$. If we denote by $\mathcal{C}_{\mathcal{U}, \mathcal{T}}$ the set of such permutation vectors, it can be verified that $\left| \mathcal{C}_{\mathcal{U}, \mathcal{T}} \right| = \lambda!t'!(\Lambda - \lambda - t')!\binom{\Lambda}{t' + \lambda}$. Hence, the subfile $W_{n, \mathcal{T}}$ is counted a total of $\lambda!t'!(\Lambda - \lambda - t')!\binom{\Lambda}{t' + \lambda}\binom{N}{K}K!/N$ times when considering the bounds $R(\bm{d}, \bm{c})$ with $\bm{d} \in \mathcal{D}_{n, \mathcal{U}}$ and $\bm{c} \in \mathcal{C}_{\mathcal{U}, \mathcal{T}}$. The same reasoning follows for each $\mathcal{U} \subseteq [\Lambda] \setminus \mathcal{T}$ with $|\mathcal{U}| = \lambda$, hence the subfile $W_{n, \mathcal{T}}$ is counted a total of
\begin{equation}
     a_{t'} = \binom{\Lambda - t'}{\lambda} \lambda!t'!(\Lambda - \lambda - t')!\binom{\Lambda}{t' + \lambda} \frac{\binom{N}{K}K!}{N}
\end{equation}
times, which gives us the number of times this same subfile appears in~\eqref{eqn: Complete Lower Bound 1}. The same reasoning follows for any $n \in [N]$ and for any $\mathcal{T} \subseteq [\Lambda]$ with $|\mathcal{T}| = t'$. Thus, the expression in~\eqref{eqn: Complete Lower Bound 1} can be rewritten as
\begin{align}
    R^\star_{\text{u}} & \geq \frac{1}{\binom{N}{K}K!\Lambda!} \sum_{t' = 0}^{\Lambda} N a_{t'} x_{t'} \\
                          & = \sum_{t' = 0}^{\Lambda} \frac{\binom{\Lambda}{t' + \lambda}}{\binom{\Lambda}{t'}} x_{t'} \\
                          & = \sum_{t' = 0}^{\Lambda} f(t') x_{t'}
\end{align}
where $f(t')$ and $x_{t'}$ are defined as
\begin{align}
    f(t') & \coloneqq \frac{\binom{\Lambda}{t' + \lambda}}{\binom{\Lambda}{t'}} \\
    0 \leq x_{t'} & \coloneqq \sum_{n \in [N]} \sum_{\mathcal{T} \subseteq [\Lambda] : |\mathcal{T}| = t'} \frac{\left| W_{n, \mathcal{T}} \right|}{NB}.
\end{align}

At this point, we seek to lower bound the minimum worst-case load $R^\star_{\text{u}}$ by lower bounding the solution to the following optimization problem
\begin{subequations}\label{eqn: Optimization Problem 1}
  \begin{alignat}{2}
    & \min_{\bm{x}}  & \quad & \sum_{t' = 0}^{\Lambda} f(t') x_{t'} \\
    & \text{subject to}  &  & \sum_{t' = 0}^{\Lambda} x_{t'} = 1 \label{eqn: File-Size Constraint} \\
    & & & \sum_{t' = 0}^{\Lambda} t' x_{t'} \leq \frac{\Lambda M}{N} \label{eqn: Memory-Size Constraint}
  \end{alignat}
\end{subequations}
where \eqref{eqn: File-Size Constraint} and \eqref{eqn: Memory-Size Constraint} correspond to the file-size constraint and the cumulative cache-size constraint, respectively.

\subsubsection{Lower Bounding the Solution to the Optimization Problem}

Since the auxiliary variable $x_{t'}$ can be considered as a probability mass function, the optimization problem in \eqref{eqn: Optimization Problem 1} can be seen as the minimization of $\mathbb{E}[f(t')]$. Moreover, the following holds.
\begin{lemma}\label{lem: Strictly Decreasing Sequence}
  The function $f(t')$ is convex and decreasing in $t'$.
\end{lemma}
\begin{proof}
    The proof is relegated to the longer version of this work~\cite{Brunero2021FundamentalLimitsCombinatorial}, where the proof is provided for a generalization of the function here denoted by $f(t')$.
\end{proof}
Taking advantage of \Cref{lem: Strictly Decreasing Sequence}, we can write $\mathbb{E}[f(t')] \geq f(\mathbb{E}[t'])$ using Jensen's inequality. Then, considering that $f(t')$ is also decreasing with increasing $t' \in [0 : \Lambda]$, we can further write $f(\mathbb{E}[t']) \geq f(\Lambda M/N)$ taking advantage of the fact that $\mathbb{E}[t']$ is upper bounded as in \eqref{eqn: Memory-Size Constraint}. Consequently, $\mathbb{E}[f(t')] \geq f(\Lambda M/N)$, and thus for $t \coloneqq \Lambda M/N$ the optimal worst-case load $R^\star_{\text{u}}$ is lower bounded by $R_{\text{LB}}$ which is a piecewise linear curve with corner points
\begin{equation}\label{eqn: Lower Bound}
  (M, R_{\text{LB}}) = \left(t\frac{N}{\Lambda}, \frac{\binom{\Lambda}{t + \lambda}}{\binom{\Lambda}{t}} \right), \quad \forall t \in [0 : \Lambda].
\end{equation}
This concludes the proof.\qed

\section{Conclusions}\label{sec: Conclusions}

In this work, we derived the fundamental limits of a coded caching scenario with exceptional potential. We proposed a novel information-theoretic converse that manages to capture the topological properties of the multi-access model in \Cref{sec: System Model}. The lower bound matches the achievable performance of the coding scheme in~\cite{Muralidhar2021MaddahAliNiesen}, hence allowing us to identify the exact optimal performance of multi-access caching with combinatorial topology under uncoded prefetching. Interestingly, our information-theoretic converse can be seen as a generalization of the lower bound in~\cite{Wan2020IndexCodingApproach}, which similarly proved the exact optimality of the MAN placement-and-delivery scheme.

As already mentioned, the longer version of this work in~\cite{Brunero2021FundamentalLimitsCombinatorial} provides a variety of interesting extensions. Indeed, the work in~\cite{Brunero2021FundamentalLimitsCombinatorial} not only provides a generalization of the combinatorial topology that allows for the coexistence of users connected to different numbers of caches, but also provides very interesting results regarding the topology-agnostic multi-access problem, offering novel lower bounds on the average worst-case performance when it is not known a priori how the $K$ users are connected to the $\Lambda$ caches in the system.

\bibliographystyle{IEEEtran}
\bibliography{references}

%\printbibliography

\end{document}